\documentclass[12pt,reqno,a4paper]{amsart}
\usepackage{amsmath,amssymb,dsfont,graphicx,cite,latexsym,
epsf,cancel,epic,eepic}
\usepackage[colorlinks=false]{hyperref}
\usepackage{mathpazo}
\newtheorem{theorem}{Theorem}

\newtheorem{lemma}{Lemma}

\newtheorem{definition}{Definition}
\newtheorem{corollary}{Corollary}
\def\beq{\begin{equation}}
\def\eeq{\end{equation}}
\def\bea{\begin{eqnarray}}
\def\eea{\end{eqnarray}}
\makeatletter
\expandafter\let\expandafter
\reset@font\csname reset@font\endcsname
\def\subeqnarray{\arraycolsep1pt
   \def\@eqnnum\stepcounter##1{\stepcounter{subequation}
       {\reset@font\rm(\theequation\alph{subequation})}}
\jot5mm     \eqnarray}

\makeatother
\newcommand{\bbR}{{\mathbb R}}
\newcommand{\bbC}{{\mathbb C}}

\newcommand{\cE}{{\mathcal E}}
\def\ep{\varepsilon}
\def\epsilon{\varepsilon}
\def\t{\widetilde}

\def\nn{\nonumber}
\def\endpf{\hfill$\square$\medskip}

\newbox\meibox
\def\placeunder#1#2#3#4{\setbox\meibox%
\vbox{\hbox{\hskip#4$\hphantom{#2}$}\hbox{$\hphantom{#1}$}}%
\vtop{\baselineskip=0pt\lineskiplimit=\baselineskip%
\lineskip=#3\hbox to \wd\meibox{\hfil\hskip#4$#2$\hfil}%
\hbox to \wd\meibox{\hfil$#1$\hfil}}}
\def\intprod{\mathbin{\hbox to 6pt{%
                 \vrule height0.4pt width5pt depth0pt
                 \kern-.4pt
                 \vrule height6pt width0.4pt depth0pt\hss}}}

\begin{document}
%%%%%%%%%%%%%%%%%%%%%%%%
%%%%%%%%%%%%%%%%%%%%%%%%
%%%%%%%%%%%%%%%%%%%%%%%%
%%%%%%%%%%%%%%%%%%%%%%%%%%%%%%%
%%%%%%%%%%%%%%%%%%%%%%%%%%%%%%%
\title[Geometry of Kahan discretizations of Hamiltonian systems. II]
{Geometry of the Kahan discretizations \\ of planar quadratic Hamiltonian systems. II.\\
Systems with a linear Poisson tensor}
%%%%%%%%%%%%%%%%%%%%%%%%%%%%%%%
%%%%%%%%%%%%%%%%%%%%%%%%%%%%%%%

%
\author{Matteo Petrera \and Yuri B. Suris }

\thanks{E-mail: {\tt  petrera@math.tu-berlin.de, suris@math.tu-berlin.de}}

\maketitle

\begin{center}
{\footnotesize{
Institut f\"ur Mathematik, MA 7-1\\
Technische Universit\"at Berlin, Str. des 17. Juni 136,
10623 Berlin, Germany
}}
\end{center}

\begin{abstract}
Kahan discretization is applicable to any quadratic vector field and produces a birational map which approximates the shift along the phase flow. For a planar quadratic Hamiltonian vector field with a linear Poisson tensor and with a quadratic Hamilton function, this map is known to be integrable and to preserve a pencil of conics. In the paper ``Three classes of quadratic vector fields for which the Kahan discretization is the root of a generalised Manin transformation'' by P. van der Kamp et al. \cite{KCMMOQ}, it was shown that the Kahan discretization can be represented as a composition of two involutions on the pencil of conics. In the present note, which can be considered as a  comment to that paper, we show that this result can be reversed. 
For a linear form $\ell(x,y)$, let $B_1,B_2$ be any two distinct points on the line $\ell(x,y)=-c$, and let $B_3,B_4$ be any two distinct points on the line $\ell(x,y)=c$. Set $B_0=\tfrac{1}{2}(B_1+B_3)$ and $B_5=\tfrac{1}{2}(B_2+B_4)$; these points lie on the line $\ell(x,y)=0$. Finally, let $B_\infty$ be the point at infinity on this line. Let $\mathfrak E$ be the pencil of conics with the base points $B_1,B_2,B_3,B_4$. Then the composition of the $B_\infty$-switch and of the $B_0$-switch on the pencil $\mathfrak E$ is the Kahan discretization of a Hamiltonian vector field $f=\ell(x,y)\begin{pmatrix}\partial H/\partial y \\ -\partial H/\partial x \end{pmatrix}$ with a quadratic Hamilton function $H(x,y)$. This birational map $\Phi_f:\mathbb C P^2\dashrightarrow\mathbb C P^2$ has three singular points $B_0,B_2,B_4$, while the inverse map $\Phi_f^{-1}$ has three singular points $B_1,B_3,B_5$.
\end{abstract}
%%%%%%%%%%%%%%%%%%%%%%%%
%%%%%%%%%%%%%%%%%%%%%%%%
%%%%%%%%%%%%%%%%%%%%%%%%

%%%%%%%%%%%%%%%%%%%%%%%%%%%%%%%
%%%%%%%%%%%%%%%%%%%%%%%%%%%%%%%
\section{Introduction}
%%%%%%%%%%%%%%%%%%%%%%%%%%%%%%%
%%%%%%%%%%%%%%%%%%%%%%%%%%%%%%%

The Kahan discretization was introduced in \cite{K} as a method applicable to any system of ordinary differential equations on $\bbR^n$ with a quadratic vector field:
\begin{equation}\label{eq: diff eq gen}
\dot{x}=f(x)=Q(x)+Bx+c,
\end{equation}
where each component of $Q:\bbR^n\to\bbR^n$ is a quadratic form, while $B\in{\rm Mat}_{n\times n}(\bbR)$ and $c\in\bbR^n$. Kahan's discretization reads as 
\begin{equation}\label{eq: Kahan gen}
\frac{\widetilde{x}-x}{2\epsilon}=Q(x,\widetilde{x})+\frac{1}{2}B(x+\widetilde{x})+c,
\end{equation}
where
\[
Q(x,\widetilde{x})=\frac{1}{2}\big(Q(x+\widetilde{x})-Q(x)-Q(\widetilde{x})\big)
\]
is the symmetric bilinear form corresponding to the quadratic form $Q$. Equation (\ref{eq: Kahan gen}) is {\em linear} with respect to $\widetilde x$ and therefore defines a {\em rational} map $\widetilde{x}=\Phi_f(x,\epsilon)$. Explicitly, one has
\beq \label{eq: Phi gen}
\t x =\Phi_f(x,\ep)= x + 2\ep \left( I - \ep f'(x) \right)^{-1} f(x),
\eeq
where $f'(x)$ denotes the Jacobi matrix of $f(x)$. 

Clearly, this map approximates the time $\epsilon$ shift along the solutions of the original differential system. Since equation (\ref{eq: Kahan gen}) remains invariant under the interchange $x\leftrightarrow\widetilde{x}$ with the simultaneous sign inversion $\epsilon\mapsto-\epsilon$, one has the {\em reversibility} property
\begin{equation}\label{eq: reversible}
\Phi^{-1}_f(x,\epsilon)=\Phi_f(x,-\epsilon).
\end{equation}
In particular, the map $f$ is {\em birational}. 

We will always set $\epsilon=1$ (and drop $\epsilon$ from notations). One can restore $\epsilon$ by the simple rescaling $f(x)\mapsto \epsilon f(x)$. For $\epsilon=1$, formula 
\eqref{eq: reversible} should be replaced by $\Phi^{-1}_f(x)=\Phi_{-f}(x)$.

In \cite{PS, PPS1, PPS2} the authors undertook an extensive study of the properties of the Kahan's method when applied to integrable systems. It was demonstrated that, in an amazing number of cases, the method preserves integrability in the sense that the map $\Phi_f(x,\epsilon)$ possesses as many independent integrals of motion as the original system $\dot x=f(x)$.
Further remarkable geometric properties of the Kahan's method were discovered in \cite{CMOQ1, CMOQ2, CMOQ4}.
\begin{theorem} {\bf\cite{CMOQ4}}
Consider a Hamiltonian vector field 
\beq\label{eq: CMOQ vector field}
f(x,y)=\ell(x,y)\begin{pmatrix} \partial H/\partial y \\ -\partial H/\partial x \end{pmatrix},
\eeq
where $\ell: \mathbb R^2\to \mathbb R$ is a linear form, and the Hamilton function $H:\bbR^n \to\bbR$ is a polynomial of degree 2,
$$
H(x,y)=\frac{1}{2}a_1x^2+a_2xy+\frac{1}{2}a_3y^2+a_4x+a_5y. 
$$
Then the map $\Phi_f(x)$ possesses the following rational integral of motion: 
\beq\label{eq: CMOQ int}
\t H(x,y)= \frac{C(x,y)}{D(x,y)},
\eeq
where $C(x,y)$, $D(x,y)$ are polynomials of degree 2 given by 
\beq\label{eq: CD}
C(x,y)=H(x,y)-\tfrac{1}{2} \gamma_2 \ell^2(x,y), \quad D(x,y)=1- \gamma_1 \ell^2(x,y),
\eeq
and
$$
\gamma_1=a_2^2-a_1a_3, \quad \gamma_2=a_3a_4^2+a_1a_5^2-2a_2a_4a_5.
$$
\end{theorem}
 
The level sets of the integral \eqref{eq: CMOQ int} are conics
\beq\label{eq: pencil}
\cE_\lambda=\big\{(x,y): C(x,y)-\lambda D(x,y)=0\big\},
\eeq
which form a linear system (a pencil). 

We assume that the curve $C(x,y)  =0$ is nonsingular. The second basis curve of the pencil, $D(x,y)  =0$, is reducible, and consists of two lines $\ell(x,y)=\pm c$, where  $c=\gamma_1^{-1/2}$. All curves $\cE_\lambda$ pass through the set of {\em base points} which is defined by $C(x,y)=D(x,y)=0$. Generically, there are four base points, counted with multiplicities. For any point $(x_0,y_0)\in\bbC^2$ different from the base points, there is a unique  curve $\cE_\lambda$ of the pencil such that $(x_0,y_0)\in\cE_\lambda$. It is defined by $\lambda=\t H(x_0,y_0)=C(x_0,y_0)/D(x_0,y_0)$. The orbit of $(x_0,y_0)$ lies on $\cE_\lambda$.

We recall the definition of involutions on conics and on pencils of conics.  
\begin{definition}
{\bf \cite{KMQ, KCMMOQ}}.

{\em 1)}
Consider a nonsingular conic $\cE$ in $\bbC^2$, and a point $B\not\in\cE$. The $B$-switch on $\mathcal{E}$ is the map $I_{\cE, B}: {\mathcal{E}}\rightarrow {\mathcal{E}}$ defined as follows: for any $P\in\cE$, the point $I_{\cE, B}(P)$ is the unique second intersection point of $\cE$ with the line $(BP)$.

{\em 2)}
Consider a pencil $\frak E=\{\cE_\lambda\}$ of conics in $\bbC^2$. The $B$-switch on $\frak E$ is the birational map $I_{\frak E,B}:\bbC^2\dashrightarrow\bbC^2$ defined as follows. For any $P\in\bbC^2$ which is not a base point of $\frak E$, let $\cE_\lambda$ be the unique curve of the pencil such that $P\in\cE_\lambda$, and set $I_{\frak E,B}(P)=I_{\cE_\lambda,B}(P)$. 
\end{definition}

The following statement is established in \cite{KCMMOQ} in the case when $H(x,y)$ is a quadratic form (a homogeneous polynomial of degree 2), that is, when $a_4=a_5=0$.

\begin{theorem}\label{Th composition} 
The Kahan map $\Phi_f$ can be represented as a composition of two involutions on the pencil $\mathfrak E$:
\beq \label{eq main}
\Phi_f = I_{\mathfrak E,B_\infty} \circ I_{\mathfrak E,B_0},
\eeq
where $B_\infty=[-\beta:\alpha:0]$ and 
$$
B_0=\rho(-\beta,\alpha), \quad \rho=\frac{1+(\beta a_4-\alpha a_5)}{\beta^2 a_1-2\alpha\beta a_2+\alpha^2 a_3},
$$ 
are two points on the line $\ell(x,y)=0$.
\end{theorem}
\begin{proof}
Direct computation.
\end{proof}

The following geometric relation between ingredients of the construction passed unnoticed in \cite{KCMMOQ}.
\begin{theorem}\label{Th geometry}
Generically, the map $\Phi_f$ considered as a birational map $\bbC P^2\dashrightarrow\bbC P^2$, has three singular points, $B_1$, $B_3$, $B_0$, lying on the lines $\ell(x,y)=-c$, $\ell(x,y)=c$, and  $\ell(x,y)=0$, respectively. Similarly, the map $\Phi_f^{-1}$ has three singular points, $B_2$, $B_4$, $B_5$, lying on the lines $\ell(x,y)=-c$,  $\ell(x,y)=c$, and  $\ell(x,y)=0$, respectively. The points $B_1,B_2,B_3,B_4$ are the base points of the pencil \eqref{eq: pencil}. The points $B_0$ and $B_5$ are related to those as follows:
\beq \label{eq B0}
B_0=\tfrac{1}{2}(B_2+B_4), \quad B_5=\tfrac{1}{2}(B_1+B_3).
\eeq
See Figure \ref{Fig1}.
\end{theorem}

%%%%%%%%%%%%%%%%%%%%%%%%%%%%%%%%%%%%%%%%%%%%%%%%%%%%%%%%%%%%%%%%%%%%
\begin{figure}
\begin{center}
\includegraphics[width=0.6\textwidth]{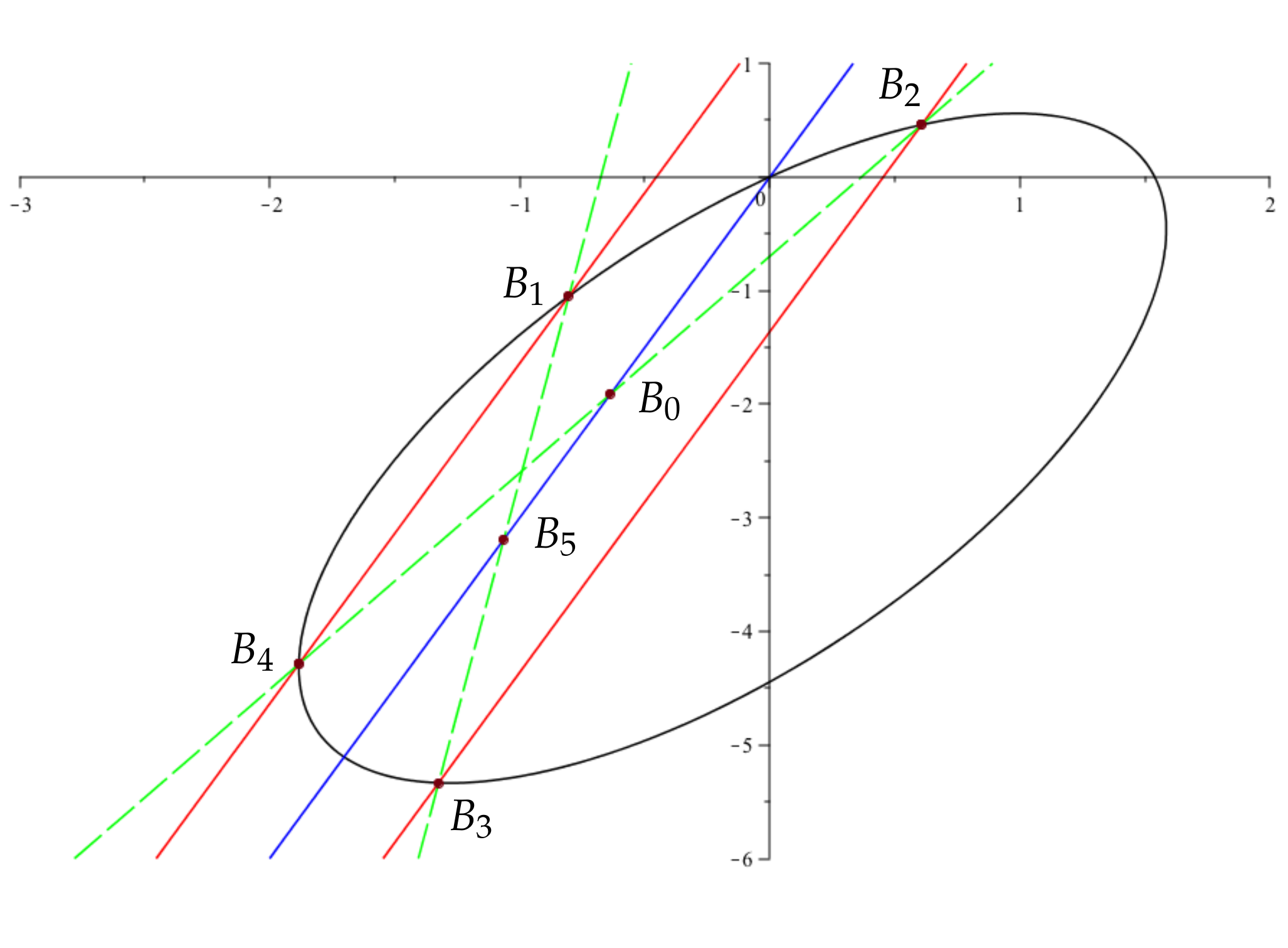}
\caption{The black curve is the conic $C=0$, the red lines represent the reducible conic $D=0$. The finite base points are $B_1, \ldots, B_4$. The points $B_0$, $B_5$ are singular points of the Kanan map $\Phi_f$, resp. of $\Phi_f^{-1}$, for the following data: $\ell(x,y)=3x-y$, $H(x,y)=-x^2+(2/5)xy+(1/2)y^2-4x+4y$.}
\label{Fig1}
\end{center}
\end{figure}
%%%%%%%%%%%%%%%%%%%%%%%%%%%%%%%%%%%%%%%%%%%%%%%%%%%%%%%%%%%%%%%%%%%

The main result we would like to add in the present paper is that Theorem \ref{Th composition} can be reversed.

\begin{theorem}\label{Th characterization}
For a linear form $\ell(x,y)=\alpha x+\beta y$, let $B_1,B_2$ be any two distinct points on the line $\ell(x,y)=-c$, let $B_3,B_4$ be any two distinct points on the line $\ell(x,y)=c$, and let $\mathfrak E$ be the pencil of conics with the base points $B_1,B_2,B_3,B_4$. Further, set 
$$
B_0= \tfrac{1}{2}(B_2+B_4).
$$ 
This point lies on the line $\ell(x,y)=0$. Let $B_\infty=[-\beta:\alpha:0]$ be the point at infinity through which this line passes. Then the map $I_{\mathfrak E,B_\infty}\circ I_{\mathfrak E,B_0}$ is the Kahan discretization $\Phi_f$ of a Hamiltonian vector field \eqref{eq: CMOQ vector field}. This map has singularities at $B_1$, $B_3$, and $B_0$.
\end{theorem}
We remark that one can take $B_0$ to be any of the four points
$$
\tfrac{1}{2}(B_1+B_3), \quad \tfrac{1}{2}(B_1+B_4), \quad \tfrac{1}{2}(B_2+B_3), \quad \tfrac{1}{2}(B_2+B_4),
$$ 
so that one obtains in this manner four different Kahan maps (or, better, two different Kahan maps along with their respective inverse maps).

%%%%%%%%%%%%%%%%%%%%%%%%%%%%%%%%%%%%%%%%
%%%%%%%%%%%%%%%%%%%%%%%%%%%%%%%%%%%%%%%%
\section{Proof of Theorem \ref{Th geometry}}
\label{sect: kahan geometry}
%%%%%%%%%%%%%%%%%%%%%%%%%%%%%%%%%%%%%%%%%
%%%%%%%%%%%%%%%%%%%%%%%%%%%%%%%%%%%%%%%%%

\begin{proof}
Due to the covariance of the Kahan discretization with respect to rotations, we can restrict ourselves to the case $\ell(x,y)=x$, so that $\alpha=1$ and $\beta=0$. 
The Hamiltonian vector field $f=\ell J\nabla H$ for the Hamilton function
\beq \label{H}
H(x,y)= \tfrac{1}{2}a_1 x^2 +a_2 xy +\tfrac{1}{2}a_3 y^2 +a_4 x+a_5 y
\eeq
is given in this case by
\bea
\dot x & = & \ell \partial H/\partial y\;=\; a_2x^2+a_3xy+a_5x,  \label{k1} \\
\dot y & = & -\ell \partial H/\partial x\;=\; -a_1x^2-a_2xy-a_4x.    \label{k2}
\eea
The Kahan discretization of this system is the map $(\t x , \t y)= \Phi_f (x,y)$ defined by the equations of motion
\bea
\t x-x  & = & 2a_2\t x x+a_3(\t xy+x \t y)+a_5(\t x+x),  \label{dk1}\\
\t y-y & = & - 2a_1\t x x-a_2(\t xy+x \t y)-a_4(\t x+x),        \label{dk2}
\eea
which can be solved for $\t x,\t y$ according to
\beq
\begin{pmatrix} \t x \\ \t y\end{pmatrix}=
A^{-1}(x,y)\begin{pmatrix}  x+ a_5 x \\  y-a_4 x
\end{pmatrix},
\eeq
where
$$
A(x,y)=\begin{pmatrix} 1-2 a_2 x-a_3 y- a_5 &  -a_3 x\\
2 a_1 x+a_2 y+ a_4  & 1+ a_2 x \end{pmatrix}.
$$
As a result, 
\beq \label{eq: Kahan map}
\t x=\frac{R(x,y)}{T(x,y)}, \quad \t y=\frac{S(x,y)}{T(x,y)},
\eeq
where $R$, $S$ and $T$ are polynomials of degree 2. They are given by
\bea
R(x,y) & = & x\big(1+a_5+(a_2+a_2a_5-a_3a_4)x+a_3y\big), \nn \\
S(x,y) & = & -2a_4x+(1-a_5)y-2(a_1+a_1a_5-a_2a_4)x^2-(3a_2+a_2a_5-a_3a_4)xy-a_3y^2, \nn \\
T(x,y) & = & 1-a_5-(a_2+a_2a_5-a_3a_4)x-a_3y-2(a_2^2-a_1a_3)x^2. \nn
\eea

Now one shows by a direct computation that the system of three equations $R(x,y)=0$, $S(x,y)=0$, $T(x,y)=0$ the singular points of $\Phi_f$ admits exactly three solutions:
$$
B_1=\Big(-c ,-\frac{1+a_5-(a_2+a_2a_5-a_3a_4)c}{a_3} \Big), \quad B_3=\Big(c ,-\frac{1+a_5+(a_2+a_2a_5-a_3a_4)c}{a_3} \Big),
$$
and 
$$
B_0=\Big(0,\frac{1-a_5}{a_3} \Big),
$$
where $c=(a_2^2-a_1a_3)^{-1/2}$. Changing the signs of all $a_k$, we find the three singular points of $\Phi_f^{-1}$:
$$
B_2=\Big(-c ,\frac{1-a_5+(a_2-a_2a_5+a_3a_4)c}{a_3} \Big), \quad B_4=\Big(c ,\frac{1-a_5-(a_2-a_2a_5+a_3a_4)c}{a_3} \Big),
$$
and 
$$
B_5=\Big(0,-\frac{1+a_5}{a_3} \Big).
$$
Equations \eqref{eq B0} follow immediately. One also confirms by a direct computation that the polynomial $C(x,y)$ vanishes at the points $B_1,B_2,B_3,B_4$.
\end{proof}

%%%%%%%%%%%%%%%%%%%%%%%%%%%%%%%%%%%%%%%%
%%%%%%%%%%%%%%%%%%%%%%%%%%%%%%%%%%%%%%%%
\section{Proof of Theorem \ref{Th characterization}}
\label{sect: kahan as gen manin}
%%%%%%%%%%%%%%%%%%%%%%%%%%%%%%%%%%%%%%%%%
%%%%%%%%%%%%%%%%%%%%%%%%%%%%%%%%%%%%%%%%%

We start with formulas for computing involutions on conics.

\begin{lemma} \label{lemma involution inf}
Let $\cE$ be a nonsingular conic in $\bbC^2$ given by the equation
\beq\label{eq: gen conic}
\cE: \quad u_1x^2+u_2xy+u_3y^2+u_4x+u_5y+u_6=0.
\eeq
Let  $B_\infty=[-\beta:\alpha: 0]$ be a point at infinity. Then the map $I_{\cE,B_\infty}: P_0=(x_0,y_0)\mapsto P_1=(x_1,y_1)$ is given by
\bea
x_1 & = & \frac{(-u_1\beta^2+u_3\alpha^2)x_0+(-u_2\beta^2+2u_3\alpha\beta)y_0-u_4\beta^2+u_5\alpha\beta}
                       {u_1\beta^2-u_2\alpha\beta+u_3\alpha^2}, \label{inv inf x1} \\
y_1 & = & \frac{(2u_1\alpha\beta-u_2\alpha^2)x_0+(u_1\beta^2-u_3\alpha^2)y_0+u_4\alpha\beta-u_5\alpha^2}
                       {u_1\beta^2-u_2\alpha\beta+u_3\alpha^2} . \label{inv inf y1} 
\eea
\end{lemma}
\begin{proof} Equation of the line $(B_\infty P_0)$ is $y=\mu x + \nu$, where $\mu=-\alpha/\beta$ and $\nu=(\alpha x_0+\beta y_0)/\beta$. To find the second intersection point $P_1=(x_1,y_1)$ of this line with $\cE$, we substitute this equation into equation $U(x,y)=0$ of the conic. We have:
\beq \label{involution quadratic eq}
U(x,\mu x+\nu)=A_2x^2+A_1x+A_0,
\eeq
where
\bea
A_2 & = & u_3\mu^2+u_2\mu+u_1, \label{A2} \\
A_1 & = & (2u_3\mu+u_2)\nu+ u_5\mu+u_4, \label{A1}\\
A_0 & = & u_3\nu^3+u_5\nu+u_6 \label{A0}.
\eea
Thus, for $x_1$ we get a quadratic equation. By the Vieta formula, we have: $x_0+x_1=-A_1/A_2$, which upon a straightforward computation gives \eqref{inv inf x1}. After that, $y_1=\mu x_1+\nu$ gives \eqref{inv inf y1}.
\end{proof}

\begin{corollary}
For a pencil $\mathfrak E$ consisting of the conics
$$
\cE_\lambda: \quad C(x,y)-\lambda\big(c^2-\ell^2(x,y)\big)=0
$$
with $\ell(x,y)=\alpha x+\beta y$, the involution $I_{\mathfrak E, B_\infty}$ with $B_\infty=[-\beta:\alpha:0]$ is an affine map of $\bbC^2$.
\end{corollary}
\begin{proof} It is sufficient to observe that the formulas for the map $I_{\cE_\lambda,B_\infty}$ do not depend on $\lambda$. Indeed, the shift $u_1\mapsto u_1+\lambda \alpha^2$, $u_2\mapsto u_2+2\lambda\alpha\beta$, $u_3\mapsto u_3+\lambda\beta^2$ leaves formulas \eqref{inv inf x1}, \eqref{inv inf y1} invariant. Moreover, these formulas do not involve $u_6$ at all.
\end{proof}

\begin{lemma} \label{lemma involution finite}
Let $\cE$ be a nonsingular conic in $\bbC^2$ given by the equation
$$
\cE: \quad u_1x^2+u_2xy+u_3y^2+u_4x+u_5y+u_6=0.
$$
Let $B_0=(x_0,y_0)\not\in\cE$. Then the map $I_{\cE,B}: P_1=(x_1,y_1)\mapsto P_2=(x_2,y_2)$ is given by
\bea
x_2 & = &  -x_1-\frac{ (2u_3 \mu +u_2)\nu +u_5 \mu +u_4}{u_3 \mu^2 +u_2 \mu +u_1}, \label{x2} \\
y_2 & = &  -y_1-\frac{-(u_2\mu+2u_1) \nu+u_5 \mu^2+u_4\mu}{u_3 \mu^2 +u_2 \mu +u_1}. \label{y2}
\eea
where
\beq
\mu = \frac{y_1 -y_0}{x_1 -x_0}, \quad \nu= y_0 - \mu x_0=y_1-\mu x_1\nn.
\eeq
\end{lemma}
\begin{proof}
Proceeding as before, we observe that in this case equation of the line $(B_0P_1)$ is $y=\mu x+\nu$. Thus, for the second intersection point $P_2=(x_2,y_2)$ of this line with $\cE$ we get a quadratic equation (\ref{involution quadratic eq}), with the coefficients $A_i$ given by \eqref{A2}--\eqref{A0}.  We compute $x_2$ from the Vieta formula $x_1+x_2=-A_1/A_2$, which leads to \eqref{x2}. To compute $y_2$, we observe that
$$
y_2=\mu x_2+\nu=-\mu\Big(x_1+\frac{A_1}{A_2}\Big)+\nu=-y_1+2\nu-\mu\frac{A_1}{A_2},
$$
which leads to \eqref{y2} after a short computation.
\end{proof}

{\em Proof of Theorem \ref{Th characterization}.} We set
$$
B_1=(-c,\xi_1), \quad B_2=(-c,\xi_2), \quad B_3=(c,\xi_3), \quad B_4=(c,\xi_4).
$$
Consider the pencil of conics with the base points $B_1,B_2,B_3,B_4$. It will have the form \eqref{eq: pencil} with a certain quadratic polynomial $C(x,y)$ and with $D(x,y)=c^2-x^2$.  For a given $(x_0,y_0)$, we determine the curve of the pencil through this point by setting 
$$
\lambda= \frac{C(x_0,y_0)}{c^2-x_0^2}.
$$ 
The equation of this curve $\cE$ is as in \eqref{eq: gen conic}. We set $B_0=\big(0,(\xi_2+\xi_4)/2\big)$, and use formulas of Lemmas \ref{lemma involution inf}, \ref{lemma involution finite} to compute $I_{\cE,B_\infty}\circ I_{\cE,B_0}$ . Notice that in the present case, $\alpha=1$, $\beta=0$, the formulas of Lemma \ref{lemma involution inf} simplify considerably and read
\beq
x_1=-x_0, \quad 
y_1  =-y_0-\frac{u_2x_0+u_5} {u_3}.
\eeq
The result of this computation is that $I_{\cE,B_\infty}\circ I_{\cE,B_0}$ is given by formulas \eqref{eq: Kahan map} with certain quadratic polynomials $R(x,y)$, $S(x,y)$, $T(x,y)$. It remains to check whether this map satisfies equations of motion \eqref{dk1}--\eqref{dk2}, which reduces to a system of linear equations for the coefficients $a_k$. A direct computation shows that this system admits a unique solution:
\bea
a_1 & = & \frac{\xi_1\xi_3+\xi_2\xi_4-\xi_1\xi_2-\xi_3\xi_4}{\xi_1+\xi_3-\xi_2-\xi_4} \ c^{-2},\\
a_2 & = & \frac{\xi_3+\xi_4-\xi_1-\xi_2}{\xi_1+\xi_3-\xi_2-\xi_4} \ c^{-1},\\
a_3 & = & \frac{-4}{\xi_1+\xi_3-\xi_2-\xi_4} ,\\
a_4 & = & \frac{\xi_1\xi_2-\xi_3\xi_4}{\xi_1+\xi_3-\xi_2-\xi_4} \ c^{-1},\\
a_5 & = & \frac{\xi_3+\xi_4+\xi_1+\xi_2}{\xi_1+\xi_3-\xi_2-\xi_4} .
\eea
This proves the theorem. \endpf

%%%%%%%%%%%%%%%%%%%%%%%%%%%%%%%
%%%%%%%%%%%%%%%%%%%%%%%%%%%%%%%
\section{Conclusions}
%%%%%%%%%%%%%%%%%%%%%%%%%%%%%%%
%%%%%%%%%%%%%%%%%%%%%%%%%%%%%%%

In \cite {PSS}, we proved an amazing characterization of integrable maps arising as Kahan discretizations of quadratic planar Hamiltonian vector fields with a constant Poisson tensor, in terms of the geometry of their set of invariant curves. Here, these results are extended to quadratic planar Hamiltonian vector fields with a linear Poisson tensor. Such a neat geometric characterization of Kahan discretizations is quite unexpected and surprizing and supports our belief expressed in \cite{PPS1, PPS2} that Kahan-Hirota-Kimura discretizations will serve as a rich source of novel results concerning algebraic geometry of integrable birational maps. It will be desirable to find similar characterizations for further classes of integrable Kahan discretizations, in dimensions $n>2$.

%%%%%%%%%%%%%%%%%%%%%%
%%%%%%%%%%%%%%%%%%%%%%
\section*{Acknowledgment}
%%%%%%%%%%%%%%%%%%%%%%
%%%%%%%%%%%%%%%%%%%%%%
This research is supported by the DFG Collaborative Research Center TRR 109 ``Discretization in Geometry and Dynamics''.

%%%%%%%%%%%%%%%%%%%%%%
%%%%%%%%%%%%%%%%%%%%%%

%%%%%%%%%%%%%%%%%%%%%%%%%%%%%%%
%%%%%%%%%%%%%%%%%%%%%%%%%%%%%%%
\end{document}